\documentclass[12pt]{iopart}
\usepackage{iopams,mathptm,times,graphicx}

\makeatletter

\newcounter{theorem}
\@addtoreset{theorem}{section}
\renewcommand\thetheorem{\arabic{section}.\arabic{theorem}}
\newenvironment{definition}{\par\medskip\noindent\begingroup{\bf Definition
             \stepcounter{theorem}\thetheorem.}\ \itshape
             \def\@currentlabel{\thetheorem}}{\endgroup\par\medskip}
\newenvironment{lemma}{\par\medskip\noindent\begingroup{\bf Lemma
             \stepcounter{theorem}\thetheorem.}\ \itshape
             \def\@currentlabel{\thetheorem}}{\endgroup\par\medskip}
\newenvironment{theorem}{\par\medskip\noindent\begingroup{\bf Theorem
             \stepcounter{theorem}\thetheorem.}\ \itshape
             \def\@currentlabel{\thetheorem}}{\endgroup\par\medskip}

\newenvironment{remark}{\par\medskip\noindent\begingroup{\bf Remark
             \stepcounter{theorem}\thetheorem.}\
             \def\@currentlabel{\thetheorem}}{\endgroup\par\medskip}
\newenvironment{corollary}{\par\medskip\noindent\begingroup{\bf Corollary
             \stepcounter{theorem}\thetheorem.}\ \itshape
             \def\@currentlabel{\thetheorem}}{\endgroup\par\medskip}
\newenvironment{proof}{\par\noindent{\bf Proof.} }{\proofbox\par\medskip}

\def\proofbox{\hfill{\ensuremath\Box}}

\newdimen\LENB \newdimen\LENW \newdimen\THI
\newdimen\LENWH \newdimen\LENTOT \newcount\N
\def\vbrknlnele#1#2#3{
  \LENB=#1pt \LENW=#2pt \THI=#3pt
  \LENWH=\LENW \divide\LENWH by 2
  \LENTOT=\LENB \advance\LENTOT by \LENW
  \vbox to \LENTOT{
    \vbox to \LENWH{}
    \nointerlineskip
    \vbox to \LENB{\hbox to \THI{\vrule width \THI height \LENB}}
    \nointerlineskip
    \vbox to \LENWH{}
  }}

\def\vbrknln#1{
  \N=#1
  \vcenter{
    \vbox{
      \loop\ifnum\N>0
        \vbox to 4pt{\vbrknlnele{2}{2}{0.1}}
        \nointerlineskip
        \advance\N by -1
      \repeat
  }}}

\def\hbrknlnele#1#2#3{
  \LENB=#1pt \LENW=#2pt \THI=#3pt
  \LENTOT=\LENB \advance\LENTOT by \LENW
  \vcenter{
    \vbox to \THI{
      \hbox to \LENTOT{
        \hfil
        \vrule width \LENB height \THI
        \hfil}
  }}}

\makeatletter

\usepackage{ulem}

\def\journal#1&#2,{\begingroup \let\journal=\dummyjournal
               \it #1\unskip~\bf\ignorespaces #2\rm,\endgroup}
\def\dummyjournal{\errmessage{Reference foul up: nested \journal macros}}

\eqnobysec

\def\eqref#1{(\ref{#1})}
\begin{document}
\title[On the $\tau$-functions of
the reduced Ostrovsky equation and the $A_2^{(2)}$ 2D-Toda system]
  {On the $\tau$-functions of
the reduced Ostrovsky equation and the $A_2^{(2)}$ two-dimensional Toda system}
\author{Bao-Feng Feng$^1$, Ken-ichi Maruno$^1$ and
Yasuhiro Ohta$^{2}$
}
\address{$^1$~Department of Mathematics,
The University of Texas-Pan American,
Edinburg, TX 78539-2999
}
\address{$^2$~Department of Mathematics,
Kobe University, Rokko, Kobe 657-8501, Japan
}
\eads{\mailto{feng@utpa.edu}, \mailto{kmaruno@utpa.edu} and
\mailto{ohta@math.kobe-u.ac.jp}}

\date{\today}
\def\submitto#1{\vspace{28pt plus 10pt minus 18pt}
     \noindent{\small\rm To be submitted to : {\it #1}\par}}

\begin{abstract}
The reciprocal link between the reduced Ostrovsky equation and the $A_2^{(2)}$
two-dimensional Toda system is used to construct
the $N$-soliton solution of the
 reduced Ostrovsky equation.
The $N$-soliton solution of the reduced Ostrovsky
equation is presented in the form of pfaffian through a hodograph
(reciprocal) transformation.
The bilinear equations and
the $\tau$-function of the reduced Ostrovsky equation are
 obtained from the period 3-reduction of the $B_{\infty}$ or
 $C_{\infty}$
two-dimensional Toda system, {\it i.e.}, the $A_2^{(2)}$ two-dimensional Toda system.
One of $\tau$-functions of the $A_2^{(2)}$ two-dimensional Toda system
becomes the square of a pfaffian which
also become a solution of the reduced Ostrovsky equation.
There is another bilinear equation which is a member of the 3-reduced
 extended BKP hierarchy. Using this bilinear equation,
we can also construct the same pfaffian solution.
\par
\kern\bigskipamount\noindent
\today
\end{abstract}

\kern-\bigskipamount

\submitto{\JPA}

\section{Introduction}
In this paper, we study the $N$-soliton
solutions of the reduced Ostrovsky equation~
\cite{Ostrovsky,Stepanyants}
\begin{equation}
\partial_x
\left(\partial_t+u\partial_x\right)u
-3u=0\,,\label{vakhnenko}
\end{equation}
which is a special case ($\beta=0$) of the Ostrovsky equation
\begin{equation}
\partial_x
\left(\partial_t+u\partial_x+\beta \partial_x^3\right)u
-\gamma u=0\,.\label{ostrovsky}
\end{equation}
The Ostrovsky equation
was originally derived as a model for weakly nonlinear surface and internal
waves in a rotating ocean~\cite{Ostrovsky,Stepanyants}.
Later, the same equation was derived from different physical situations
by several authors~\cite{Stepanyants,Hunter,Vakhnenko1}.
Note that the reduced Ostrovsky equation (\ref{vakhnenko}) is sometimes
called the Vakhnenko equation or the Ostrovsky-Hunter equation
~\cite{Vakhnenko1,Vakhnenko2,Vakhnenko3,Vakhnenko4,Liu}.
Vakhnenko et al. constructed the $N$ (loop) soliton
solution of the reduced
Ostrovsky equation by using a hodograph (reciprocal) transformation and
the Hirota bilinear method~\cite{Vakhnenko2,Vakhnenko3}.
The same problem was approached from the point of view of inverse
scattering method~\cite{Vakhnenko4}.

Differentiating the reduced Ostrovsky equation (\ref{vakhnenko}) with
respect to $x$,
we obtain
\begin{equation}
u_{txx}+3u_xu_{xx}+uu_{xxx}-3u_x=0\,,\label{short-DP-eq}
\end{equation}
which is known as the short wave limit of the Degasperis-Procesi (DP)
equation~\cite{Hone-Wang,MatsunoPLA}.
This equation is derived from
the DP equation~\cite{DP}
\begin{equation}
U_T+3\kappa^3U_X- U_{TXX}+4UU_X=3U_XU_{XX}+UU_{XXX}\,,\label{DP-eq}
\end{equation}
by taking  a short wave limit $\epsilon\to 0$ with
$U=\epsilon^2(u+\epsilon u_1+\cdots)$, $T=\epsilon^{-1} t$,
$X=\epsilon x$.
Using this connection,
Matsuno constructed the $N$-soliton solution of the short wave model of
the DP equation, {\it i.e.}, the reduced Ostrovsky equation,
from the $N$-soliton solution of the DP
equation~\cite{MatsunoPLA,Matsuno-DP1,Matsuno-DP2}.
This $N$-soliton formula is equivalent to the one obtained by
Vakhnenko. Hone and Wang pointed out that there is a reciprocal link
between the reduced Ostrovsky equation and the first negative flow in
the Sawada-Kotera
hierarchy~\cite{Hone-Wang}.

In this paper, we show
the reciprocal link between
the reduced Ostrovsky equation and the $A_2^{(2)}$
two-dimensional Toda system and investigate their $\tau$-functions.
Using this reciprocal link,
we construct the $N$-soliton solution of the reduced Ostrovsky
equation in the form of pfaffian.
The bilinear equations and
the $\tau$-functions of the reduced Ostrovsky equation are
systematically obtained from the period 3-reduction of the $B_{\infty}$ or
$C_{\infty}$ two-dimensional Toda system.
One of the $\tau$-functions of the $A_2^{(2)}$ two-dimensional Toda system
becomes the square of a pfaffian, by which $N$-soliton solution 
of the reduced Ostrovsky equation is expressed.
We also show that another bilinear equation which is a 
member of the 3-reduced
extended BKP hierarchy, or the so-called negative Sawada-Kotera
hierarchy, 
can give rise to the 
reduced Ostrovsky equation by a hodograph transformation. Using this 
bilinear equation,
we can also construct the same pfaffian solution.

\section{The reduced Ostrovsky equation, the period 3-reduction of the
 $B_{\infty}$
2D-Toda system, and
 the 3-reduced extended BKP hierarchy}

\subsection{The two-dimensional Toda system of $B_{\infty}$-type and
its period 3-reduction}

The two-dimensional Toda (2D-Toda) system of $A_{\infty}$-type,
which is also called the Toda field equation  or the two-dimensional
Toda lattice,
is given as follows~\cite{Leznov-Saveliev,Mikhailov1,Mikhailov2,Mikhailov3}:
\begin{equation}
\frac{\partial^2\theta_n}{\partial x_1 \partial x_{-1}}
=-\sum_{m\in \mathbb{Z}}a_{n,m}e^{-\theta_m}\,,
\quad n\in \mathbb{Z}\label{2dToda}
\end{equation}
where the matrix $A=(a_{n,m})$ is the transpose of the Cartan matrix for
the infinite dimensional Lie algebra $A_{\infty}$~\cite{Nimmo-Willox}.

The $A_\infty$ 2D-Toda system (\ref{2dToda}) may be written as
\begin{equation}
\frac{\partial^2\theta_n}{\partial x_1 \partial x_{-1}}
=e^{-\theta_{n-1}}-2e^{-\theta_{n}}+e^{-\theta_{n+1}}\,.
\end{equation}
The $A_\infty$
2D-Toda system (\ref{2dToda}) is transformed
into
the bilinear equation
\begin{equation}
-\left(\frac{1}{2}D_{x_1}D_{x_{-1}}-1\right)\tau_n\cdot \tau_n=\tau_{n-1}\tau_{n+1}\,,\label{2dtl-bilinear}
\end{equation}
through the dependent variable transformation
\begin{equation}
\theta_n=-\ln \frac{\tau_{n+1}\tau_{n-1}}{\tau_n^2}\,.
\end{equation}
Here $D_{x}$ is the Hirota $D$-operator which is defined
as
\begin{equation}
D_x^na(x)\cdot b(x)=\left(\partial_x-\partial_{x'}\right)^na(x)b(x')|_{x=x'}\,.
\end{equation}

\begin{lemma}[Ueno-Takasaki\cite{Ueno-Takasaki},
Babich-Matveev-Sall\cite{Matveev},
Hirota\cite{HirotaBook},Nimmo-Willox\cite{Nimmo-Willox}]
The bilinear equations of the 2D-Toda lattice hierarchy including
(\ref{2dtl-bilinear}) have the
following Gram-type determinant solution:
\begin{equation}
\tau_n={\rm det}\left(
\psi_{i,j}^{(n)}\right)_{1\leq i,j \leq M}\,\,,
\end{equation}
where
\begin{equation}
\psi_{i,j}^{(n)}=c_{i,j}+(-1)^n\int_{-\infty}^{x_1}\varphi_i^{(n)}
\hat{\varphi}_j^{(-n)}dx_1\,.
\end{equation}
Here $c_{i,j}$ are constants,
$\varphi_{i,j}^{(n)}$ and $\hat{\varphi}_{i,j}^{(n)}$ satisfy
\begin{equation}
\frac{\partial \varphi_i^{(n)}}{\partial x_k}=\varphi_i^{(n+k)}\,,\quad
\frac{\partial \hat{\varphi}_i^{(n)}}{\partial x_k}=(-1)^{k-1}\hat{\varphi}_i^{(n+k)}\,,
\end{equation}
for $k=\pm 1, \pm 2, \pm 3, \cdots$.

For example, the linear independent set of functions
$\{\varphi_i^{(n)},\hat{\varphi}_j^{(n)}\}$ such that
$$
\varphi_i^{(n)}= p_i^ne^{\xi_i}\,,\quad
\hat{\varphi}_j^{(n)}=q_j^ne^{\eta_j}\,,
$$
$\xi_i=p_ix_1+\frac{1}{p_i}x_{-1}+p_i^2x_2+\frac{1}{p_i^2}x_{-2}+p_i^3x_3
+\frac{1}{p_i^3}x_{-3}+\cdots +\xi_{i0}$ and
$\eta_i=q_ix_1+\frac{1}{q_i}x_{-1}
-q_i^2x_2-\frac{1}{q_i^2}x_{-2}+q_i^3x_3
+\frac{1}{q_i^3}x_{-3}+\cdots +\eta_{i0}$ for $i,j=1,2,\cdots, M$,
gives $M$-soliton solution of the $A_{\infty}$ 2D-Toda system.
\end{lemma}
\begin{proof}
See \cite{HirotaBook}.
\end{proof}

We impose the $B_{\infty}$-reduction $\theta_{n}=\theta_{1-n}$ ($n\geq
1$) to the $A_{\infty}$ 2D-Toda system (\ref{2dToda}),
{\it i.e.},
fold the infinite sequence $\{\dots,
 \theta_{-2},\theta_{-1},\theta_0,\theta_1,\theta_2,\dots\}$
in the midpoint between $\theta_0$ and
$\theta_1$~\cite{Jimbo-Miwa,Ueno-Takasaki,Nimmo-Willox,Willox}.
Under this constraint, we have $\theta_0=\theta_1$, $\theta_{-1}=\theta_2$,
$\theta_{-2}=\theta_3$, \dots .

For $n=1$,
\begin{eqnarray*}
\frac{\partial^2\theta_1}{\partial x_1 \partial x_{-1}}
&=&e^{-\theta_{0}}-2e^{-\theta_{1}}+e^{-\theta_{2}}\\
&=&  e^{-\theta_{2}} -e^{-\theta_{1}}\\
&=& e^{-\theta_{2}} -2e^{-(\theta_{1}+\ln 2)}\,.
\end{eqnarray*}
For $n=2$,
\begin{eqnarray*}
\frac{\partial^2\theta_2}{\partial x_1 \partial x_{-1}}
&=&e^{-\theta_{1}}-2e^{-\theta_{2}}+e^{-\theta_{3}}\\
&=&2e^{-(\theta_{1}+\ln 2)}-2e^{-\theta_{2}}+e^{-\theta_{3}}\,.
\end{eqnarray*}
After redefining $\theta_1$ by $\theta_1+\ln 2\to \theta_1$, we obtain
the $B_{\infty}$ 2D-Toda system~\cite{Leznov-Saveliev,Mikhailov3}:
\begin{equation}
\frac{\partial^2\theta_n}{\partial x_1 \partial x_{-1}}
=-\sum_{m\in \mathbb{Z}_{\geq 1}}a_{n,m}e^{-\theta_m}\,,\quad {\rm for}\quad n\in \mathbb{Z}_{\geq 1}\,,\label{2dToda-B}
\end{equation}
where the matrix $A=(a_{n,m})$ is the transpose of the Cartan matrix for
the infinite dimensional Lie algebra $B_{\infty}$~\cite{Nimmo-Willox}.

The $B_{\infty}$ 2D-Toda system (\ref{2dToda-B})
is transformed into the bilinear equations
\begin{eqnarray}
&&-\left(\frac{1}{2}D_{x_1}D_{x_{-1}}-1\right)\tau_1\cdot
 \tau_1=\tau_{1}\tau_{2}\,,\label{BToda-bilinear1}\\
&&-\left(\frac{1}{2}D_{x_1}D_{x_{-1}}-1\right)\tau_n\cdot
 \tau_n=\tau_{n-1}\tau_{n+1}
\,, \quad {\rm for} \quad n\geq 2\,,\label{BToda-bilinear2}
\end{eqnarray}
through the dependent variable transformation
\begin{equation}
\theta_1=-\ln \frac{\tau_2}{\tau_1}\,,\qquad {\rm and} \qquad
\theta_n=-\ln \frac{\tau_{n+1}\tau_{n-1}}{\tau_n^2}\quad {\rm
for}\,\,\,n\geq 2.
\end{equation}

\begin{lemma}
The bilinear equations of the $B_{\infty}$ 2D-Toda system
(\ref{BToda-bilinear1}) and (\ref{BToda-bilinear2})
have
the $N$-soliton solution which is expressed as
\begin{equation}
\tau_n={\rm det}\left(
\psi_{i,j}^{(n)}\right)_{1\leq i,j \leq 2N}\,\,,
\end{equation}
where
\begin{eqnarray}
&&\psi_{i,j}^{(n)}=c_{i,j}-2(-1)^n\int_{-\infty}^{x_1}\varphi_i^{(n)}
\varphi_j^{(-n+1)}dx_1\,,\\
&&\varphi_i^{(n)}= p_i^ne^{\xi_i}\,,\quad
\xi_i=p_ix_1+\frac{1}{p_i}x_{-1}+p_i^3x_3
+\frac{1}{p_i^3}x_{-3}+\cdots +\xi_{i0}\,,\nonumber
\end{eqnarray}
and $c_{i,j}=-c_{j,i}$, $c_{i,i}=0$.
\end{lemma}
\begin{proof}
Imposing the $B_{\infty}$ reduction $\tau_n=\tau_{1-n}$, {\it i.e.}, folding the
sequence of the $\tau$-functions $\{\dots,
 \tau_{-2},\tau_{-1},\tau_0,\tau_1,\tau_2, \tau_3,\dots\}$
in the midpoint between $\tau_0$ and $\tau_1$, we have
$\tau_0=\tau_1$, $\tau_{-1}=\tau_2$, $\tau_{-2}=\tau_3$, ....
~\cite{Jimbo-Miwa,Ueno-Takasaki,Nimmo-Willox,Willox}.
Thus we obtain the bilinear equations (\ref{BToda-bilinear1})
and (\ref{BToda-bilinear2}) from the $A_\infty$
2D-Toda bilinear equation
 (\ref{2dtl-bilinear}).

To impose the $B_{\infty}$ reduction to the Gram-type determinant
 solution of the $A_{\infty}$ 2D-Toda system,
we impose the constraint
$\hat{\varphi}_j^{(n)}=-2\frac{\partial}{\partial x_1}\varphi_j^{(n)}
=-2\varphi_j^{(n+1)}$, $c_{i,j}=-c_{j,i}$,
$c_{i,i}=0$, $M=2N$ and $x_{2k}\equiv 0$ for every interger $k$.
With this constraint,
each element of the Gram-type determinant has the
following property:
\begin{eqnarray}
\fl
\psi_{i,j}^{(n)}&=&c_{i,j}-2(-1)^n\int_{-\infty}^{x_1}\varphi_i^{(n)}
\varphi_j^{(-n+1)}dx_1
=-c_{j,i}+2(-1)^{1-n}
\int_{-\infty}^{x_1}\varphi_j^{(-n+1)}\varphi_i^{(n)}dx_1\nonumber\\
\fl &=&-\psi_{j,i}^{(1-n)}\,.
\end{eqnarray}
Then the $\tau$-function satisfies
$\tau_n=\tau_{1-n}$.
Therefore the $N$-soliton solution of the $B_{\infty}$ 2D-Toda lattice
is expressed by the above Gram-type determinant.
\end{proof}

\begin{definition}
Let $A=(a_{i,j})_{1\leq i,j\leq 2N}$ be a $2N\times 2N$ skew-symmetric
 matrix.
The pfaffian of $A$ is defined by
\begin{eqnarray*}
\fl {\rm pf}(A)&=&{\rm pf}(a_{i,j})_{1\leq i,j\leq 2N}
=
{\rm pf}(1,2,...,2N)
=\sum_{
{\small \begin{array}{c}
i_1<i_2<\cdots <i_N\\
i_1<j_1,...,i_N<j_N
\end{array}
}
}{\rm sgn}
(\pi)
a_{i_1,j_1}a_{i_2,j_2}
\cdots a_{i_N,j_N}\,\,,
\end{eqnarray*}
where $\pi=\left(
{\small
\begin{array}{ccccc}
1 & 2 & \cdots & 2N-1 & 2N\\
i_1 & j_1 & \cdots & i_N & j_N
\end{array}
}
\right)$
is a permutation of $\{1,2,...,2N-1,2N\}$.
\end{definition}
The pfaffian can be computed recursively by
\[
\fl {\rm pf}(1,2,...,2N-1,2N)=\sum_{i=2}^{2N}(-1)^i
{\rm pf}(1,i){\rm pf}(2,3,...,i-1,i+1,...,2N-1,2N)\,.
\]
For a skew-symmetric matrix $A$, we have
$[{\rm pf}(A)]^2={\rm det}(A)$.

\begin{lemma}\label{lemma:BToda-pfaff}
The $\tau$-function
$\tau_1$ of the bilinear equations (\ref{BToda-bilinear1}) and
(\ref{BToda-bilinear2}) is written in the form of pfaffian
\begin{equation}
\tau_1=\tau^2\,,\qquad \tau= {\rm pf}(1,2,...,2N-1,2N)\,,
\end{equation}
where ${\rm pf}(i,j)=c_{i,j}+\int_{-\infty}^{x_1}D_{x_1}\varphi_i^{(0)}
\cdot \varphi_j^{(0)}dx_1$,
$\varphi_i^{(n)}= p_i^ne^{\xi_i}$,
$\xi_i=p_i^{-1}x_{-1}+p_ix_1+p_i^3x_3+\frac{1}{p_i^3}x_{-3}+\cdots+\xi_i^0$,
$c_{i,j}=-c_{j,i}$, $c_{i,i}=0$.
\end{lemma}
\begin{proof}
See \cite{HirotaBook,Nimmo-Ruijsenaars}.
\end{proof}

\begin{lemma}\label{lemma:3BToda}
The bilinear equations
\begin{eqnarray}
&&-\left(\frac{1}{2}D_{x_1}D_{x_{-1}}-1\right)\tau_{1} \cdot
 \tau_{1}=\tau_{1}\tau_{2}\,,\label{3BToda-bilinear1}\\
&&-\left(\frac{1}{2}D_{x_1}D_{x_{-1}}-1\right)\tau_{2}\cdot \tau_{2}=\tau_1^2\,,\label{3BToda-bilinear2}
\end{eqnarray}
have
the $N$-soliton solution which is expressed as
\begin{equation}
\tau_n={\rm det}\left(
\psi_{i,j}^{(n)}\right)_{1\leq i,j \leq 2N}\,\,,
\end{equation}
where
\begin{eqnarray}
&&\psi_{i,j}^{(n)}=c_{i,j}+\frac{2p_i}{p_i+p_j}\left(
-\frac{p_{i}}{p_j}\right)^{n-1}
e^{\xi_i+\xi_j}
\,,\qquad \xi_i=p_ix_1+\frac{1}{p_i}x_{-1}+\xi_{i0}\,,\nonumber
\end{eqnarray}
and $c_{i,j}=\delta_{j,2N+1-i}c_i$\,, $c_i=-c_{2N+1-i}$\,,
$p_i^2-p_ip_{2N+1-i}+p_{2N+1-i}^2=0$\,.
\end{lemma}
\begin{proof}
Imposing a period 3-reduction
$\tau_n=\tau_{n+3}$ to the sequence of the $\tau$-functions
$\{\dots ,\tau_3
 ,\tau_2,\tau_1,\tau_1,\tau_2,\tau_3,\dots\}$,
it becomes $\{\dots ,\tau_1
 ,\tau_2,\tau_1,\tau_1,\tau_2,\tau_1,\dots\}$~
\cite{Jimbo-Miwa,Ueno-Takasaki,Nimmo-Willox,Willox}.
Thus we obtain two bilinear equations (\ref{3BToda-bilinear1}) and
(\ref{3BToda-bilinear2}).

In Lie algebraic terms, this corresponds to the reduction
to the affine Lie algebra $A_2^{(2)}$ from the infinite
 dimensional
Lie algebra
 $B_{\infty}$~\cite{Jimbo-Miwa,Ueno-Takasaki,Nimmo-Willox,Willox}.

To impose a period 3-reduction to the $N$-soliton solution,
we add a constraint $p_i^3+p_{2N+1-i}^3=0$
($p_i\neq -p_{2N+1-i}$) for $i=1,2,\cdots , N$
and $c_{i,j}=\delta_{j,2N+1-i}\,c_i$, $c_i=-c_{2N+1-i}$.
To show $\tau_{n+3}=\tau_n$, we manipulate $\tau_n$ as follows:
\begin{eqnarray*}
\fl
\tau_n&=&{\rm det}\left(\delta_{j,2N+1-i}c_i+\frac{2p_i}{p_i+p_j}\left(
-\frac{p_i}{p_j}\right)^{n-1}e^{\xi_i+\xi_j}\right)_{1\leq i,j\leq 2N}\\
\fl &=&e^{2\sum_{k=1}^{2N}\xi_k}\,
{\rm det}\left(\delta_{j,2N+1-i}c_i\left(
-\frac{p_j}{p_i}\right)^{n-1}e^{-\xi_i-\xi_j}+
\frac{2p_i}{p_i+p_j}\right)_{1\leq
i,j\leq 2N}\\
\fl &=&e^{2\sum_{k=1}^{2N}\xi_k}\,
{\rm det}\left(\delta_{j,2N+1-i}c_i\left(
-\frac{p_{2N+1-i}}{p_i}\right)^{n-1}
e^{-\xi_i-\xi_{2N+1-i}}+\frac{2p_i}{p_i+p_j}\right)_{1\leq
i,j\leq 2N}\,.
\end{eqnarray*}
Since $\left(
-\frac{p_{2N+1-i}}{p_i}\right)^3=1$,
the relation $\tau_{n+3}=\tau_n$ is satisfied.
\end{proof}

Letting $f=\tau_{1}$ and $g=\tau_{2}$,
we obtain the bilinear equations
\begin{eqnarray}
&&-\left(\frac{1}{2}D_{x_1}D_{x_{-1}}-1\right)f\cdot
 f=fg\,,\label{bilinear-1}\\
&&-\left(\frac{1}{2}D_{x_1}D_{x_{-1}}-1\right)g\cdot g=f^2\,.
\label{bilinear-2}
\end{eqnarray}

\begin{corollary}\label{corollary:tau}
The $\tau$-function
$f=\tau_1$ of the bilinear equations (\ref{bilinear-1}) and
(\ref{bilinear-2}) is written in the form of pfaffian
\begin{equation}
f=\tau_1=\tau^2\,,\qquad \tau= {\rm pf}(1,2,...,2N-1,2N)\,,
\end{equation}
where
\begin{eqnarray}
&&{\rm pf}(i,j)=c_{i,j}
+\frac{p_i-p_j}{p_i+p_j}e^{\xi_i+\xi_{j}}
\,,\qquad \xi_i=p_ix_1+\frac{1}{p_i}x_{-1}+\xi_{i0}\,,\nonumber
\end{eqnarray}
and $c_{i,j}=\delta_{j,2N+1-i}c_i$\,, $c_i=-c_{2N+1-i}$\,,
$p_i^2-p_ip_{2N+1-i}+p_{2N+1-i}^2=0$\,.
\end{corollary}
\begin{proof}
Using Lemma \ref{lemma:BToda-pfaff} and \ref{lemma:3BToda}, we obtain
 the above formula.
\end{proof}

Let $k_i=p_i+p_{2N+1-i}$. From
$p_i^2-p_ip_{2N+1-i}+p_{2N+1-i}^2=0$, we obtain
$p_i=\frac{1}{6}(3k_i+{\rm i}\sqrt{3}k_i)$,
$p_{2N+1-i}=\frac{1}{6}(3k_i-{\rm i}\sqrt{3}k_i)$,
$p_ip_{2N+1-i}=\frac{k_i^2}{3}$ and
$\frac{1}{p{i}}+\frac{1}{p_{2N+1-i}}=\frac{3}{k_i}$.
Thus
\[
\xi_i+\xi_{2N+1-i}
=k_ix_1+\frac{3}{k_i}x_{-1}+\xi_{i0}+\xi_{{2N+1-i}0}\,.
\]
In the pfaffian solution, all phase functions can be expressed by the
summation of $\xi_i+\xi_{2N+1-i}$. So the phase functions can be
expressed by the parameters $\{k_i\}$ ($i=1,2,\cdots,N$).
Each coefficient of exponential functions can be normalized to 1 after
absorption into phase constants or can be rewritten by the parameters
$\{k_i\}$.
Thus it is possible to rewrite the above $\tau$-function by using the
parameters $\{k_i\}$ instead of $\{p_i\}$.\vspace{0.2in}

\noindent{\bf Examples}:\\
For $N=1$,
\begin{equation}
\fl \tau={\rm pf}(1,2)=c_1+\frac{p_1-p_2}{p_1+p_2}e^{\xi_1+\xi_{2}}\,.
\end{equation}
Letting $c_1=1$ and
$e^{\gamma_1}=\frac{p_1-p_2}{p_1+p_2}$,
the $\tau$-function can be rewritten as
\begin{equation}
\fl \tau={\rm pf}(1,2)=1+e^{\xi_1+\xi_{2}+\gamma_1}\,,
\end{equation}
where
\[
\xi_1+\xi_{2}
=k_1x_1+\frac{3}{k_1}x_{-1}+\xi_{10}+\xi_{{2}0}\,.
\]

\noindent
For $N=2$,
\begin{eqnarray*}
\fl \tau&=&{\rm pf}(1,2,3,4)
={\rm pf}(1,2){\rm pf}(3,4)-{\rm pf}(1,3){\rm pf}(2,4)
+{\rm pf}(1,4){\rm pf}(2,3)\\
\fl &=&\frac{p_1-p_2}{p_1+p_2}e^{\xi_1+\xi_{2}}
 \frac{p_3-p_4}{p_3+p_4}e^{\xi_3+\xi_{4}}
-\frac{p_1-p_3}{p_1+p_3}e^{\xi_1+\xi_{3}}
 \frac{p_2-p_4}{p_2+p_4}e^{\xi_2+\xi_{4}}\\
\fl &&\qquad +\left(c_1+\frac{p_1-p_4}{p_1+p_4}e^{\xi_1+\xi_{4}}\right)
\left(c_2+\frac{p_2-p_3}{p_2+p_3}e^{\xi_2+\xi_{3}}\right)\\
\fl &=&c_1c_2+c_2\frac{p_1-p_4}{p_1+p_4}e^{\xi_1+\xi_{4}}
+c_1\frac{p_2-p_3}{p_2+p_3}e^{\xi_2+\xi_{3}}\\
\fl &&\quad +\left(\frac{p_1-p_2}{p_1+p_2}
\frac{p_3-p_4}{p_3+p_4}
-\frac{p_1-p_3}{p_1+p_3}\frac{p_2-p_4}{p_2+p_4}
+\frac{p_1-p_4}{p_1+p_4}\frac{p_2-p_3}{p_2+p_3}
\right)e^{\xi_1+\xi_{2}+\xi_3+\xi_{4}}\,.
\end{eqnarray*}
Letting $c_1=c_2=1$ and $e^{\gamma_1}=\frac{p_1-p_4}{p_1+p_4}$ and
$e^{\gamma_2}=\frac{p_2-p_3}{p_2+p_3}$, the above $\tau$-function
becomes
\begin{eqnarray*}
\tau&=&1+e^{\xi_1+\xi_{4}+\gamma_1}
+e^{\xi_2+\xi_{3}+\gamma_2}
+b_{12}e^{\xi_1+\xi_{2}+\xi_3+\xi_{4}+\gamma_1+\gamma_2}\,,
\end{eqnarray*}
where
\begin{eqnarray*}
\fl b_{12}&=&
\frac{p_1-p_2}{p_1+p_2}
\frac{p_3-p_4}{p_3+p_4}
\frac{p_1+p_4}{p_1-p_4}\frac{p_2+p_3}{p_2-p_3}
-\frac{p_1-p_3}{p_1+p_3}\frac{p_2-p_4}{p_2+p_4}
\frac{p_1+p_4}{p_1-p_4}\frac{p_2+p_3}{p_2-p_3}
+1\\
\fl &=&\frac{(k_1-k_2)^2(k_1^2-k_1k_2+k_2^2)}
{(k_1+k_2)^2(k_1^2+k_1k_2+k_2^2)}\,,
\end{eqnarray*}
\[
\fl \xi_1+\xi_{4}
=k_1x_1+\frac{3}{k_1}x_{-1}+\xi_{10}+\xi_{{4}0}\,,\quad
\xi_2+\xi_{3}
=k_2x_1+\frac{3}{k_2}x_{-1}+\xi_{20}+\xi_{{3}0}\,.
\]

The $N$-soliton solution is written in the following form:
\begin{eqnarray}
\fl &&\qquad \tau=\sum_{\mu=0,1}\exp\left[\sum_{i=1}^N\mu_i\eta_i
+\sum_{i<j}^{(N)}\mu_i\mu_j\ln b_{ij} \right]\,,\\
\fl &&\qquad  b_{ij}=
\frac{(k_i-k_j)^2(k_i^2-k_ik_j+k_j^2)}{(k_i+k_j)^2(k_i^2+k_ik_j+k_j^2)}\,
\quad {\rm for}\quad
i<j\,,
\quad \eta_i=\xi_i+\xi_{2N+1-i}
=k_ix_1+\frac{3}{k_i}x_{-1}+\eta_{i0}\,,\nonumber
\end{eqnarray}
where $\displaystyle \sum_{\mu=0,1}$ means the summation over all possible
combinations of $\mu_i=0$ or $1$ for $i=1,2,\cdots,N$,
and $\displaystyle \sum_{i<j}^{(N)}$ means the summation over all possible
combinations
of $N$ elements under the condition $i<j$.
This is consistent with the results in \cite{Vakhnenko2,Vakhnenko3}.

\begin{lemma}\label{lemma:hodograph}
The $\tau$-function $f$ of (\ref{bilinear-1}) and
 (\ref{bilinear-2})
gives the solution of the reduced Ostrovsky equation
through the dependent variable transformation
\[
u=-(\ln f)_{x_1x_1}=-2(\ln \tau)_{x_1x_1}\,,
\]
and the hodograph (reciprocal) transformation
\begin{equation}
\left\{
\begin{array}{c}
x=x_{-1}+\int_{-\infty}^{x_1}u(x_1',x_{-1})dx_1'\\
\quad =x_{-1}-(\ln f)_{{x_1}}\,,\qquad \qquad \\
t=x_1 \,.\qquad \qquad \qquad \qquad
\end{array}
\right.\label{hodograph}
\end{equation}
\end{lemma}

\begin{proof}
The bilinear equation (\ref{bilinear-1}) is rewritten as
\begin{equation}
-(\ln f)_{x_1x_{-1}}=\frac{g}{f}-1\,.\label{bilinear-1-2}
\end{equation}
Let
\begin{equation}
\rho=\frac{f}{g}=\frac{1}{1-(\ln f)_{x_1x_{-1}}}\,,\quad
u=-(\ln f)_{x_1x_{1}}\,.
\end{equation}
Then (\ref{bilinear-1-2}) becomes
\begin{equation}
u_{x_{-1}}=\left(\frac{1}{\rho}\right)_{x_1}\,.\label{u-rho-relation}
\end{equation}
This is rewritten as
\begin{equation}
(\ln \rho)_{x_1}=-\rho u_{x_{-1}}\,.\label{rho-u}
\end{equation}

The bilinear equations (\ref{bilinear-1}) and (\ref{bilinear-2})
are written as
\begin{eqnarray}
&&-(\ln f)_{x_1x_{-1}}+1=\frac{1}{\rho}\,,\label{ln-bi-1}\\
&&-(\ln g)_{x_1x_{-1}}+1=\rho^2\,.\label{ln-bi-2}
\end{eqnarray}
Subtracting (\ref{ln-bi-2}) from (\ref{ln-bi-1}), we obtain
\begin{equation}
-(\ln \rho)_{x_1x_{-1}}=\frac{1}{\rho}-\rho^2\,,\label{tzitzeica}
\end{equation}
which leads to
\begin{equation}
\rho^3=1+\rho (\ln \rho)_{x_{1}x_{-1}}\,.
\end{equation}
Using (\ref{rho-u}), it becomes
\begin{equation}
\rho^3=1-\rho (\rho u_{x_{-1}})_{x_{-1}}\,.\label{rho-relation}
\end{equation}

Let us consider the hodograph (reciprocal) transformation
\begin{equation}
\left\{
\begin{array}{c}
x=x_{-1}+\int_{-\infty}^{x_1}u(x_1',x_{-1})dx_1'\\
\quad =x_{-1}-(\ln f)_{{x_1}}\,,\qquad \qquad \\
t=x_1 \,.\qquad \qquad \qquad \qquad
\end{array}
\right.\label{hodograph}
\end{equation}
This yield
\begin{equation}
\left\{
\begin{array}{c}
\frac{\partial x}{\partial x_{-1}}=1-(\ln
 f)_{x_1x_{-1}}\,,\\
\frac{\partial x}{\partial x_{1}}=-(\ln f)_{x_1x_{1}}\,,\quad \\
\end{array}
\right.
\end{equation}
and
\begin{equation}
\left\{
\begin{array}{c}
\partial_{x_{-1}}=\frac{1}{\rho}\partial_x\,,\quad\\
\partial_{x_{1}}=\partial_t+u\partial_x\,.\\
\end{array}
\right.\label{hodograph-derivative}
\end{equation}
Applying the hodograph (reciprocal)
transformation to (\ref{rho-u}) and
(\ref{rho-relation}),
we obtain
\begin{equation}
\left\{
\begin{array}{c}
(\partial_t+u\partial_x)\ln\rho=-u_x\,,\\
\rho^3=1-u_{xx}\,.\qquad \quad
\end{array}
\right.\label{vakhnenko-form2}
\end{equation}
This is equivalent to
\begin{equation}
(\partial_t+u\partial_x)\ln(1-u_{xx})=-3u_x\,,\label{short-DP-eq3}
\end{equation}
which can be written as
\begin{equation}
(\partial_t+u\partial_x)(1-u_{xx})=-3u_x(1-u_{xx})\,.\label{short-DP-eq2}
\end{equation}
This is nothing but the short wave limit of the DP equation
which is equivalent to the reduced Ostrovsky equation.
\end{proof}

\begin{remark}
Setting $\rho=\frac{1}{\mathcal{H}}$ in (\ref{tzitzeica}),
we obtain
\begin{equation}
(\ln \mathcal{H})_{x_1x_{-1}}=\mathcal{H}-\frac{1}{\mathcal{H}^2}\,\,.
\end{equation}
This is the Tzitzeica equation which is one of
 prime integrable systems that appears in classical differential geometry
~\cite{Tzitzeica1,Tzitzeica2,Willox,Nimmo-Ruijsenaars}.
This equation is sometimes called the Dodd-Bullough-Mikhailov equation
 since this was found independently by Dodd-Bullough and
Mikhailov~\cite{Dodd-Bullough,Mikhailov1,Mikhailov2}.
\end{remark}

\begin{theorem}\label{theorem:solution}
The $N$-soliton solution of the reduced Ostrovsky equation
is given by the following formula:
\[
u=-(\ln f)_{x_1x_1}=-2(\ln \tau)_{x_1x_1}\,,
\]
\begin{equation}
f=\tau^2\,,\qquad \tau= {\rm pf}(1,2,...,2N-1,2N)\,,
\end{equation}
where
\begin{eqnarray}
&&{\rm pf}(i,j)=c_{i,j}
+\frac{p_i-p_j}{p_i+p_j}e^{\xi_i+\xi_{j}}
\,,\qquad \xi_i=p_ix_1+\frac{1}{p_i}x_{-1}+\xi_{i0}\,,\nonumber
\end{eqnarray}
and $c_{i,j}=\delta_{j,2N+1-i}c_i$\,, $c_i=-c_{2N+1-i}$\,,
$p_i^2-p_ip_{2N+1-i}+p_{2N+1-i}^2=0$\,,
\begin{equation}
\left\{
\begin{array}{c}
x=x_{-1}+\int_{-\infty}^{x_1}u(x_1',x_{-1})dx_1'\\
\quad =x_{-1}-(\ln f)_{{x_1}}\,,\qquad \qquad \\
t=x_1 \,.\qquad \qquad \qquad \qquad
\end{array}
\right.
\end{equation}
\end{theorem}
\begin{proof}
From Corollary \ref{corollary:tau} and Lemma \ref{lemma:hodograph},
we obtain this theorem.
\end{proof}

\begin{corollary}\label{corollary:solution}
The $N$-soliton solution of the reduced Ostrovsky equation
is given by the following formula:
\[
u=-2(\ln \tau)_{x_1x_1}\,,
\]
\begin{eqnarray}
\fl &&\qquad \tau=\sum_{\mu=0,1}\exp\left[\sum_{i=1}^N\mu_i\eta_i
+\sum_{i<j}^{(N)}\mu_i\mu_j\ln b_{ij} \right]\,,\\
\fl &&\qquad  b_{ij}=
\frac{(k_i-k_j)^2(k_i^2+k_ik_j+k_j^2)}{(k_i+k_j)^2(k_i^2-k_ik_j+k_j^2)}\,
\quad {\rm for}\quad
i<j\,,
\quad \eta_i=k_ix_1+\frac{3}{k_i}x_{-1}+\eta_{i0}\,,\nonumber
\end{eqnarray}
and
\begin{equation}
\left\{
\begin{array}{c}
x=x_{-1}+\int_{-\infty}^{x_1}u(x_1',x_{-1})dx_1'\\
\quad =x_{-1}-(\ln f)_{{x_1}}\,,\qquad \qquad \\
t=x_1 \,.\qquad \qquad \qquad \qquad
\end{array}
\right.
\end{equation}
Here $\displaystyle \sum_{\mu=0,1}$ means the summation over all possible
combinations of $\mu_i=0$ or $1$ for $i=1,2,\cdots,N$,
and $\displaystyle \sum_{i<j}^{(N)}$ means the summation over all possible
combinations
of $N$ elements under the condition $i<j$.
\end{corollary}
Note that the form of the $N$-soliton solution in
Corollary \ref{corollary:solution}
is equivalent to the one obtained by
Morrison, Parkes and Vakhnenko in \cite{Vakhnenko3}.

\begin{remark}
It is well known that the 3-reduced B-Toda system is exactly
 the same as the 3-reduced C-Toda system. Thus the above $\tau$-function
 can be obtained by the period 3-reduction of
$C_{\infty}$ 2D-Toda system. For more details, see Appendix.
\end{remark}

\subsection{The reduced Ostrovsky equation and
the 3-reduced extended BKP hierarchy}

Consider a member of the extended BKP hierarchy (the BKP hierarchy with
negative time variables)
~\cite{HirotaBook,Jimbo-Miwa,DJKM1,DJKM2,DJKM3,Hirota-BKP1,Hirota-BKP2}
\begin{equation}
[(D_{x_3}-D_{x_1}^3)D_{x_{-1}}+3D_{x_1}^2]\tau\cdot \tau=0\,,\label{BKP-bilinear}
\end{equation}
which is the bilinear equation of a nonlinear partial differential equation
\begin{equation}
w_{x_{-1}x_3}-w_{x_1x_1x_1x_{-1}}-3(w_{x_1}w_{x_{-1}})_{x_1}+3w_{x_1x_1}=0\,,
\end{equation}
through the dependent variable transformation
$w=2(\ln \tau)_{x_1}$~\cite{HirotaBook}.

The $N$-soliton solution of (\ref{BKP-bilinear}) is expressed
in the form of pfaffian.
\begin{lemma}\label{lemma:BKP-tau}
[Date-Jimbo-Kashiwara-Miwa~\cite{DJKM1},
Hirota~\cite{HirotaBook,Hirota-BKP1,Hirota-BKP2}]
The bilinear equation (\ref{BKP-bilinear})
has the pfaffian solution
\[
\tau={\rm pf}(1,2,...,2N-1,2N)\,,
\]
where ${\rm pf}(i,j)=c_{i,j}+\int_{-\infty}^{x_1}D_{x_1}h_i({\bf x})\cdot
h_j({\bf x})dx_1$, $c_{i,j}=-c_{j,i}$, $c_{i,i}=0$,
${\bf x}=(\dots, x_{-3},x_{-1},x_1,x_3,x_5\dots)$
and $h_i({\bf x})$ satisfies the linear
differential equation
\begin{equation}
\frac{\partial}{\partial x_n}h_i({\bf x})=
\frac{\partial^n}{\partial x_1^n}h_i({\bf x})\,.
\end{equation}
Here we define
\begin{equation}
\frac{\partial}{\partial x_{-1}}h_i({\bf x})\equiv
\int_{-\infty}^{x_{1}}h_i({\bf x})dx_1\,,
\end{equation}
for $n=-1$. Particularly, we could choose
\begin{equation}
h_i({\bf x})=\exp(p_i^{-1}x_{-1}+p_ix_1+p_i^3x_3+\xi_i^0)\,.
\end{equation}
\end{lemma}

Alternatively, the $N$-soliton solution of (\ref{BKP-bilinear}) can be expressed as
\begin{equation}
\tau=\sum_{\mu_i=0,1} \exp
\left[
\sum_{i=1}^N\mu_i\eta_i+\sum_{1\leq i<j\leq N}B_{i,j}\mu_i\mu_j
\right]\,,
\end{equation}
where $\sum_{\mu_i=0,1}$ denotes the summation over all possible
combinations of $\mu_i=0,1$ for $i=1,2,\cdots,N$,
and $\sum_{1\leq i<j\leq N}$ is the sum over
all pairs $i,j\,\, (i<j)$ chosen from $\{1,2,\cdots,2N\}$.
Here $\exp(\eta_i)$ is defined as
\[
 \exp(\eta_i)=\exp(\xi_i+\hat{\xi}_i)\,,
\]
where
$\xi_i=p_i^{-1}x_{-1}+p_ix_1+p_i^3x_3+\xi_i^0$ and
$\hat{\xi}_i=q_i^{-1}x_{-1}+q_ix_1+q_i^3x_3+\hat{\xi}_i^0$,
and $B_{i,j}$ is given by
\[
e^{B_{i,j}}=b_{i,j}=\frac{(p_i-p_j)(p_i-q_j)(q_i-p_j)(q_i-q_j)}
{(p_i+p_j)(p_i+q_j)(q_i+p_j)(q_i+q_j)}\,.
\]

\begin{proof}
See \cite{HirotaBook}.
\end{proof}

\begin{lemma}\label{lemma:pfaffian-BKP}
The bilinear equation
\begin{equation}
(D_{x_{-1}}D_{x_1}^3-3D_{x_1}^2)\tau \cdot \tau=0\,,
\end{equation}
has the the following $N$-soliton solution:
\[
\tau={\rm pf}(1,2,...,2N-1,2N)\,,
\]
where
\begin{eqnarray}
&&{\rm pf}(i,j)=c_{i,j}
+\frac{p_i-p_j}{p_i+p_j}e^{\xi_i+\xi_{j}}
\,,\qquad \xi_i=p_ix_1+{p^{-1}_i}x_{-1}+\xi_{i0}\,,\nonumber
\end{eqnarray}
and
$c_{i,j}=\delta_{j,2N+1-i}c_i$\,, $c_i=-c_{2N+1-i}$\,,
$p_i^2-p_ip_{2N+1-i}+p_{2N+1-i}^2=0$.
\end{lemma}

\begin{proof}
Since all phase functions are expressed by the summation of
$\xi_i+\xi_{2N+1-i}$, $x_3$ is eliminated from the $\tau$-function by imposing constraints: 
$p_i^2-p_ip_{2N+1-i}+p_{2N+1-i}^2=0$ for $i=1,2,\cdots , N$ and
$c_{i,j}=\delta_{j,2N+1-i}c_i$, $c_i=-c_{2N+1-i}$ to the $N$-soliton
solution of the BKP $\tau$-function in Lemma \ref{lemma:BKP-tau}~
\cite{Jimbo-Miwa,DJKM1,DJKM2,DJKM3}. Thus $\partial_{x_3}\tau=0$ is satisfied.
\end{proof}

\begin{lemma}\label{lemma:hodograph-BKP}
The reduced Ostrovsky equation (\ref{vakhnenko}) is transformed into
the bilinear equation
\begin{equation}
(D_{x_{-1}}D_{x_1}^3-3D_{x_1}^2)\tau \cdot \tau=0\,,
\end{equation}
through the dependent transformation
\begin{equation}
u=-w_{x_1}=-2(\ln \tau)_{x_1x_1}\,,
\end{equation}
and the hodograph (reciprocal) transformation
(\ref{hodograph}).
\end{lemma}
\begin{proof}
Applying (\ref{hodograph-derivative})
to the reduced Ostrovsky equation (\ref{vakhnenko}),
we obtain
\begin{equation}
\rho u_{x_1x_{-1}}=3u\,,
\end{equation}
which yields
\begin{equation}
u_{x_1x_{-1}}=\frac{3}{\rho}u\,.\label{vakhnenko-transform}
\end{equation}
Let us introduce $w$ such that $u=-w_{x_1}=-(\ln f)_{x_1x_{1}}$.
From (\ref{ln-bi-1}), we obtain
\begin{equation}
\frac{1}{\rho}=1-w_{x_{-1}}\,.
\end{equation}
Using this relation, (\ref{vakhnenko-transform}) can be rewritten as
\begin{equation}
w_{x_1x_1x_{-1}}+3w_{x_{-1}}w_{x_1}-3w_{x_1}=0\,.
\end{equation}
This equation is transformed into a bilinear equation
\begin{equation}
(D_{x_{-1}}D_{x_1}^3-3D_{x_1}^2)\tau \cdot \tau=0\,,
\label{specialIto-bilinear}
\end{equation}
through the dependent transformation
\begin{equation}
w=2(\ln \tau)_{x_1}\,.
\end{equation}
Note that the hodograph (reciprocal) transformation (\ref{hodograph})
can be expressed as
\begin{equation}
\left\{
\begin{array}{c}
x=x_{-1}-2(\ln \tau)_{x_1}\,,\\
t=x_1 \,.\qquad \qquad \quad
\end{array}
\right.\label{hodograph-B}
\end{equation}
\end{proof}

From Lemma \ref{lemma:pfaffian-BKP} and \ref{lemma:hodograph-BKP}, we
can also prove Theorem \ref{theorem:solution}.

\begin{remark}
Vakhnenko and Parkes found the hodograph (reciprocal)
transformation (\ref{hodograph-B}) and the bilinear equation
(\ref{specialIto-bilinear}). 
It is noted that we can switch $x_{-1}$ and $x_{1}$, i.e., the 
alternative hodograph (reciprocal) transformation 
$x=x_{1}-(\ln f)_{x_{-1}}=x_{1}-2(\ln \tau)_{x_{-1}}$, $t=x_{-1}$ can be
 used instead of the above hodograph (reciprocal) transformation.
\end{remark}

\begin{remark}
Morrison, Parkes and Vakhnenko pointed out that the bilinear equation
 (\ref{specialIto-bilinear}) is a special case of
the Ito equation~\cite{Vakhnenko3,Ito}.
Hone and Wang pointed out that the
 reduced Ostrovsky equation is linked into the first negative flow of the
 Sawada-Kotera hierarchy.
The bilinear equation
 (\ref{specialIto-bilinear}) is nothing but the bilinear equation
of the first negative flow of the Sawada-Kotera hierarchy
which can be obtained from the BKP hierarchy with the 3-reduction,
{\it i.e.}, this is related to the $A_{2}^{(2)}$ affine
Lie algebra~\cite{Jimbo-Miwa}.
\end{remark}

\begin{remark}
The relationship between the Tzitzeica equation and the 3-reduced BKP
 hierarchy was pointed out by Lambert et al.~\cite{Lambert}.
\end{remark}

\section{Conclusions}

In this paper, we have shown
the reciprocal link between
the reduced Ostrovsky equation and the $A_2^{(2)}$
2D-Toda system.
Using this reciprocal link, we have constructed
the $N$-soliton solution of the reduced Ostrovsky
equation in the form of pfaffian.
The bilinear equations and
the $\tau$-function of the reduced Ostrovsky equation have been
obtained from the period 3-reduction of the $B_{\infty}$ or
$C_{\infty}$ 2D-Toda system.
One of the $\tau$-functions of the $A_2^{(2)}$ 2D-Toda system
has been shown to be the square of a pfaffian, by which the $N$-soliton 
solution of the reduced Ostrovsky equation is represented.
We have shown that the bilinear equation for a member of the 3-reduced 
extended BKP hierarchy
(the first negative flow of the Sawada-Kotera hierarchy) 
can also give rise to the 
the reduced Ostrovsky equation, thus, the same pfaffian solution.

It should be noted that a similar approach in this paper
was used to the short wave model of the
Camassa-Holm equation in our previous paper~\cite{discreteSCH}.
We proposed an integrable semi-discrete and an integrable
fully discrete analogues of
the short wave model of the Camassa-Holm equation
based on bilinear equations and a determinant solution.
Thus we can construct integrable discrete analogues of
the reduced Ostrovsky equation (or the short wave model of the
DP equation) by using the result in this paper.
We also would like to comment on the $N$-soliton solution of
the DP equation can be obtained by the reduction of the 2D-Toda system developed in this paper.
Our approach gives the $N$-soliton solution of the
DP equation in a much simpler way
compared to Matsuno's method~\cite{Matsuno-DP1,Matsuno-DP2}.
We will report the detail in our forthcoming paper.

\section*{Appendix: The period 3-reduction of
 the $C_{\infty}$ 2D-Toda system}

In this appendix, we show that the $\tau$-function obtained in the
section 2 can be obtained from the period
3-reduction of the $C_{\infty}$ 2D-Toda system.

We impose the $C_{\infty}$-reduction $\theta_{n}=\theta_{-n}$ ($n\geq
0$) to the $A_{\infty}$ 2D-Toda system, {\it i.e.},
fold the infinite sequence $\{\dots,
 \theta_{-2},\theta_{-1},\theta_0,\theta_1,\theta_2,\dots\}$
in
$\theta_{0}$~\cite{Jimbo-Miwa,Ueno-Takasaki,Nimmo-Willox,Willox}.
Then we obtain
the $C_{\infty}$ 2D-Toda system~\cite{Leznov-Saveliev,Mikhailov3}:
\begin{equation}
\frac{\partial^2\theta_n}{\partial x_1 \partial x_{-1}}
=-\sum_{m\in \mathbf{Z}_{\geq 0}}a_{n+1\,,\,m+1}e^{-\theta_m}\,,
\quad n\in \mathbf{Z}_{\geq 0}\,,\label{2dToda-C}
\end{equation}
where the matrix $A=(a_{n,m})$ is the transpose of the Cartan matrix for
the infinite dimensional Lie algebra $C_{\infty}$~\cite{Nimmo-Willox}.

The $C_{\infty}$ 2D-Toda system (\ref{2dToda-C})
is transformed into the bilinear equations
\begin{eqnarray}
&&-\left(\frac{1}{2}D_{x_1}D_{x_{-1}}-1\right)\tau_{0}\cdot
 \tau_{0}=\tau_{1}^2\,,\label{2dtl-bilinear-C-1}\\
&&-\left(\frac{1}{2}D_{x_1}D_{x_{-1}}-1\right)\tau_n\cdot
 \tau_n=\tau_{n-1}\tau_{n+1}
\,, \quad {\rm for} \quad n\geq 1\,,\label{2dtl-bilinear-C-2}
\end{eqnarray}
through the dependent variable transformation
\begin{equation}
\theta_{0}=-\ln \frac{\tau_1^2}{\tau_{0}^2}\,,\qquad {\rm and}
\qquad
\theta_n=-\ln \frac{\tau_{n+1}\tau_{n-1}}{\tau_n^2}\quad {\rm
for}\,\,\,\,n\geq 1.
\end{equation}

\begin{lemma}
The bilinear equations (\ref{2dtl-bilinear-C-1}) and (\ref{2dtl-bilinear-C-2})
have the $N$-soliton solution which is expressed as
\begin{equation}
\tau_n={\rm det}\left(
\psi_{i,j}^{(n)}\right)_{1\leq i,j \leq 2N}\,\,,
\end{equation}
where
\begin{eqnarray}
&&\psi_{i,j}^{(n)}=c_{i,j}+(-1)^n\int_{-\infty}^{x_1}\varphi_i^{(n)}
\varphi_j^{(-n)}dx_1\,,\\
&&\varphi_i^{(n)}=p_i^ne^{\xi_i}\,,
\qquad \xi_i=p_ix_1+\frac{1}{p_i}x_{-1}+p_i^3x_3
+\frac{1}{p_i^3}x_{-3}+\cdots +\xi_{i0}\,,\nonumber
\end{eqnarray}
and $c_{i,j}=c_{j,i}$.
\end{lemma}
\begin{proof}
Imposing the $C_{\infty}$ reduction $\tau_n=\tau_{-n}$, {\it i.e.},
folding the
sequence of the $\tau$-functions $\{\dots,
 \tau_{-2},\tau_{-1},\tau_0,\tau_1,\tau_2,\dots\}$
in
$\tau_{0}$, we have
$\tau_{-1}=\tau_1$, $\tau_{-2}=\tau_2$, $\tau_{-3}=\tau_3$, ....
~\cite{Jimbo-Miwa,Ueno-Takasaki,Nimmo-Willox,Willox}.
Thus we obtain the bilinear equations (\ref{2dtl-bilinear-C-1}) and
(\ref{2dtl-bilinear-C-2}) from the 2D-Toda bilinear equation
 (\ref{2dtl-bilinear}).

To impose the $C_{\infty}$ reduction to the Gram-type determinant
 solution of the $A_{\infty}$ 2D-Toda system,
we impose the constraint
$\hat{\varphi}_j^{(n)}=\varphi_j^{(n)}$, $c_{i,j}=c_{j,i}$,
$M=2N$ and $x_{2k}\equiv 0$ for every interger $k$.
With this constraint,
each element of the Gram-type determinant has the
following property:
\begin{eqnarray}
\fl \psi_{i,j}^{(n)}&=&c_{i,j}+(-1)^n\int_{-\infty}^{x_1}\varphi_i^{(n)}
\varphi_j^{(-n)}dx_1
=c_{j,i}+(-1)^{-n}\int_{-\infty}^{x_1}\varphi_j^{(-n)}
\varphi_i^{(n)}dx_1\nonumber\\
\fl &=&\psi_{j,i}^{(-n)}\,.
\end{eqnarray}
Then the $\tau$-function satisfies
$\tau_n=\tau_{-n}$.
Therefore the $N$-soliton solution of the $C_{\infty}$ 2D-Toda system
is expressed by the above Gram-type determinant.
\end{proof}

\begin{lemma}
The bilinear equations
\begin{eqnarray}
&&-\left(\frac{1}{2}D_{x_1}D_{x_{-1}}-1\right)\tau_{0}\cdot
 \tau_{0}=\tau_1^2\,,\label{3CToda-bilinear1}\\
&&-\left(\frac{1}{2}D_{x_1}D_{x_{-1}}-1\right)\tau_{1} \cdot
 \tau_{1}=\tau_{1}\tau_{0}\,,
\label{3CToda-bilinear2}
\end{eqnarray}
have
the $N$-soliton solution which is expressed as
\begin{equation}
\tau_n={\rm det}\left(
\psi_{i,j}^{(n)}\right)_{1\leq i,j \leq 2N}\,\,,
\end{equation}
where
\begin{eqnarray}
\fl &&\psi_{i,j}^{(n)}=c_{i,j}+(-1)^n\int_{-\infty}^{x_1}\varphi_i^{(n)}
\varphi_j^{(-n)}dx_1\,,\qquad
\varphi_i^{(n)}=p_i^ne^{\xi_i}\,,
\quad \xi_i=p_ix_1+\frac{1}{p_i}x_{-1}
+\xi_{i0}\,,\nonumber
\end{eqnarray}
and $c_{i,j}=\delta_{j,2N+1-i}\alpha_i$, $p_i^2-p_ip_{2N+1-i}+p_{2N+1-i}^2=0$.
\end{lemma}
\begin{proof}
Imposing a period 3-reduction
$\tau_n=\tau_{n+3}$ to the sequence of the
 $\tau$-functions
$\{\dots ,\tau_3, \tau_2 ,\tau_1, \tau_0,\tau_1,\tau_2,\tau_3,\dots\}$,
it becomes
$\{\dots ,\tau_0, \tau_1 ,\tau_1
 ,\tau_0,\tau_1,\tau_1,\tau_0,\dots\}$.
Thus we obtain two bilinear equations (\ref{3CToda-bilinear1}) and
(\ref{3CToda-bilinear2}).

In Lie algebraic terms, this corresponds to the reduction
to the affine Lie algebra
 $A_2^{(2)}$ from the infinite
 dimensional Lie algebra $C_{\infty}$.

To impose a period 3-reduction to the $N$-soliton solution,
we add a constraint $p_i^3+p_{2N+1-i}^3=0$,
($p_i\neq -p_{2N+1-i}$) for $i=1,2,\cdots , N$
and $c_{i,j}=\delta_{j,2N+1-i}\alpha_i$.
To show $\tau_{n+3}=\tau_n$, we manipulate $\tau_n$ as follows:
\begin{eqnarray*}
\tau_n&=&{\rm det}\left(\delta_{j,2N+1-i}\,\alpha_i
+\frac{1}{p_i+p_j}\left(
-\frac{p_i}{p_j}\right)^ne^{\xi_i+\xi_j}\right)_{1\leq i,j\leq 2N}\\
&=&e^{2\sum_{k=1}^{2N}\xi_k}\,
{\rm det}\left(\delta_{j,2N+1-i}\,\alpha_i\left(
-\frac{p_j}{p_i}\right)^ne^{-\xi_i-\xi_j}+\frac{1}{p_i+p_j}\right)_{1\leq
i,j\leq 2N}\\
&=&e^{2\sum_{k=1}^{2N}\xi_k}\,
{\rm det}\left(\delta_{j,2N+1-i}\,\alpha_i\left(
-\frac{p_{2N+1-i}}{p_i}\right)^n
e^{-\xi_i-\xi_{2N+1-i}}+\frac{1}{p_i+p_j}\right)_{1\leq
i,j\leq 2N}\\
\end{eqnarray*}
Since $\left(
-\frac{p_{2N+1-i}}{p_i}\right)^3=1$,
$\tau_{n+3}=\tau_n$ is satisfied.
\end{proof}

\begin{lemma}\label{lemma:pfaffian-C}
The $\tau$-function
$\tau_1$ of the bilinear equations (\ref{3CToda-bilinear1}) and
(\ref{3CToda-bilinear2}) is written in the form of pfaffian
\begin{equation}
\tau_1=\frac{1}{2^{2N}\prod_{k=1}^{2N}p_k}
\,\tau^2\,,\qquad \tau= {\rm pf}(1,2,...,2N-1,2N)\,,
\end{equation}
where ${\rm pf}(i,j)=c_{i,j}+\frac{p_i-p_j}{p_i+p_j}e^{\xi_i+\xi_j}$,
$\xi_i=p_i^{-1}x_{-1}+p_ix_1+\xi_i^0$,
$c_{i,j}=\delta_{j,2N+1-i}c_i$\,,
$c_i=-c_{2N+1-i}$\,, $p_i^2-p_ip_{2N+1-i}+p_{2N+1-i}^2=0$.
\end{lemma}
\begin{proof}
Let $\alpha_i=\alpha_{2N+1-i}$.
For $n=1$,
\begin{eqnarray*}
\tau_1&=&
{\rm det}\left(\psi_{i,j}^{(1)}\right)_{1\leq i,j\leq 2N}
={\rm det}\left(\delta_{j,2N+1-i}\,\alpha_i-\frac{1}{p_i+p_j}
\frac{p_i}{p_j}e^{\xi_i+\xi_j}\right)_{1\leq i,j\leq 2N}\\
&=&\frac{1}{2^{2N}\prod_{k=1}^{2N}p_k}\,{\rm
 det}\left(2\delta_{j,2N+1-i}\,\alpha_i\frac{p_j^2}{p_i}
-\frac{2p_j}{p_i+p_j}e^{\xi_i+\xi_j}\right)_{1\leq
 i,j\leq 2N}
\\
&=&\frac{1}{2^{2N}\prod_{k=1}^{2N}p_k}\,{\rm
 det}\left(2\delta_{j,2N+1-i}\,\alpha_i\frac{p_{2N+1-i}^2}{p_i}
-\frac{2p_j}{p_i+p_j}e^{\xi_i+\xi_j}\right)_{1\leq
 i,j\leq 2N}
\\
&=&\frac{1}{2^{2N}\prod_{k=1}^{2N}p_k}\,{\rm
 det}\left(c_{i,j}
+\frac{p_i-p_j}{p_i+p_j}e^{\xi_i+\xi_j}-e^{\xi_i+\xi_j}
\right)_{1\leq
 i,j\leq 2N}
\,,
\end{eqnarray*}
where $c_{i,j}=2\delta_{j,2N+1-i}\,\alpha_i\frac{p_{2N+1-i}^2}{p_i}$.
Then we note
\[
c_{i,j}=2\delta_{j,2N+1-i}\,\alpha_i\,\frac{p_{2N+1-i}^2}{p_i}=-2
\delta_{i,2N+1-j}\,\alpha_{2N+1-i}\,\frac{p_i^2}{p_{2N+1-i}}
=-c_{j,i}\,.
\]
Introducing $c_i=2\alpha_i\frac{p_{2N+1-i}^2}{p_i}$,
we can write as $c_{i,j}=\delta_{j,2N+1-i}c_i$ and $c_i=-c_{2N+1-i}$.

Since the $2N \times 2N$ matrix $\left(c_{i,j}
+\frac{p_i-p_j}{p_i+p_j}e^{\xi_i+\xi_j}\right)_{1\leq i,j\leq 2N}$
is skew-symmetric,
we obtain
\[
 \tau_1=\frac{1}{2^{2N}\prod_{k=1}^{2N}p_k}[{\rm pf}(1,2,...,2N-1,2N)]^2\,,
\]
where ${\rm pf}(i,j)=c_{i,j}+\frac{p_i-p_j}{p_i+p_j}e^{\xi_i+\xi_j}$.
Here we used the formula~\cite{HirotaBook}
\[
{\rm det}(a_{i,j}-y_iy_j)_{1\leq i,j\leq 2N}
={\rm det}(a_{i,j})_{1\leq i,j\leq 2N}=[{\rm pf}(1,2,...,2N)]^2\,,
\]
where $a_{i,j}=-a_{j,i}$, $a_{i,i}=0$, ${\rm pf}(i,j)=a_{i,j}$.
\end{proof}
Letting $f=\tau_{1}$ and $g=\tau_{0}$,
we obtain the bilinear equations (\ref{bilinear-1}) and
(\ref{bilinear-2}) from the bilinear equations
(\ref{3CToda-bilinear2}) and
(\ref{3CToda-bilinear1}).
As is shown in \cite{Nimmo-Willox}, the period 3-reduction of
$C_{\infty}$ 2D-Toda system gives the same result
 as the period 3-reduction of $B_{\infty}$ 2D-Toda after the relabelling
 indexes.

%%%%%%%%%%%%%%%%%%%%%%%%%%%%%%%%%%%%%%%%%%%%%%%%%%%%%%%%%%%%%%%%%%%%%%%%%

\end{document}